\documentclass{article} 
\usepackage{graphicx}           
\usepackage{amssymb}
\usepackage{amsfonts}
\usepackage{amsmath}
\usepackage{amsthm}
\usepackage{makeidx}
\usepackage{mathtools}

\usepackage{enumerate}

\newtheorem{theorem}{Theorem}
\newtheorem{proposition}{Proposition}
\newtheorem{corollary}{Corollary}
\theoremstyle{definition}
\newtheorem{example}{Example}

\newcommand{\scong}{\cong_s}
\newcommand{\id}{\text{id}}

\begin{document}
   \begin{center}
      \Large\textbf{Conjugacy of one-dimensional one-sided cellular automata is undecidable \footnote{Research supported by the Academy of Finland Grant 296018.}}\\
\vspace{0.15cm}
      \large{Joonatan Jalonen\footnote{Author's research supported by the Finnish Cultural Foundation.} and Jarkko Kari}\\
\vspace{0.15cm}
\large{University of Turku}
   \end{center}

%

\begin{abstract}
Two cellular automata are strongly conjugate if there exists a shift-commuting conjugacy between them. We prove that the following two sets of pairs $(F,G)$ of one-dimensional one-sided cellular automata over a full shift are recursively inseparable:
\begin{itemize}
\item[(i)] pairs where $F$ has strictly larger topological entropy than $G$, and
\item[(ii)]  pairs that are strongly conjugate and have zero topological entropy.
\end{itemize}

Because there is no factor map from a lower entropy system to a higher entropy one, and there is no embedding of a higher entropy system into a lower entropy system, we also get as corollaries that the following decision problems are undecidable: Given two one-dimensional one-sided cellular automata $F$ and $G$ over a full shift: Are $F$ and $G$ conjugate? Is $F$ a factor of $G$? Is $F$ a subsystem of $G$? All of these are undecidable in both strong and weak variants (whether the homomorphism is required to commute with the shift or not, respectively). It also immediately follows that these results hold for one-dimensional two-sided cellular automata.
\end{abstract}

\section{Introduction}
The original setting for cellular automata theory was the theory of computation and computability, as cellular automata were created as a mathematical model of natural computational devices. Thus algorithmic questions have always been a significant part of the study of cellular automata. It is known, for example, that surjectivity and injectivity (and so also reversibility) are decidable for one-dimensional cellular automata and undecidable in higher dimensions, and that nilpotency and equicontinuity are undecidable for one- and higher-dimensional cellular automata.

The Curtis-Lyndon-Hedlund Theorem, which says that the classical definition of cellular automata is equivalent to saying that cellular automata are shift commuting endomorphisms of the full shift, prompted the fruitful study of cellular automata as topological dynamical systems. One natural question then is to determine if two cellular automata are conjugate dynamical systems.

Combining both views, one ends up asking if conjugacy of cellular automata is decidable. In \cite{Epperlein2017} it was conjectured that topological conjugacy of one-dimensional cellular automata is undecidable. We prove that this holds for strong and weak conjugacy (whether the conjugacy is required to be shift commuting or not, respectively). In fact we prove a stronger result: Consider sets of pairs $(F,G)$ of  one-dimensional one-sided cellular automata over a full shift such that
\begin{itemize}
\item[(i)] $F$ has strictly larger topological entropy than $G$,
\item[(ii)] $F$ and $G$ are strongly conjugate and both have zero topological entropy.
\end{itemize}
We prove that these sets of pairs are recursively inseparable. The same result then also holds for one-dimensional two-sided cellular automata, too. As an immediate corollary we get that (strong) conjugacy, being a (strong) factor, and being a (strong) subsystem are undecidable properties for one-dimensional one- and two-sided  cellular automata.

\section{Preliminaries}

\subsection{Symbolic dynamics}

Zero is considered a natural number, i.e., $0\in\mathbb{N}$. For two integers $i,j\in\mathbb{Z}$ such that $i < j$ the interval from $i$ to $j$ is denoted $[i,j]=\{i,i+1,\dots,j\}$, we also denote $[i,j)=\{i,i+1,\dots j-1\}$ and $(i,j]=\{i+1,\dots,j\}$. Notation $\mathbb{M}$ is used when it does not matter whether we use $\mathbb{N}$ or $\mathbb{Z}$. Composition of functions $f:X\rightarrow Y$ and $g:Y\rightarrow Z$ is written as $gf$ and $(gf)(x)=g(f(x))$ for all $x\in X$.

The set of infinite sequences over an \emph{alphabet} $A$ indexed by $\mathbb{M}$ is $A^\mathbb{M}$. An element $c\in A^\mathbb{M}$ is a \emph{configuration}. A configuration is a function $\mathbb{M}\rightarrow A$ and we denote $c(i)=c_i$ for $i\in\mathbb{M}$. For any $D\subset\mathbb{M}$ we denote by $c_{D}$ the restriction of $c$ to the domain $D$ and by $A^D$ the set of all functions $D\rightarrow A$. The set of \emph{finite words} is denoted by $A^+=\bigcup_{n\in\mathbb{N}}A^{[0,n]}$. Let $D$ be finite and $u\in A^D$, then we denote $[u]=\{c\in A^\mathbb{M}\mid c_D=u\}$ and call such sets \emph{cylinders}. Let $A$ have the discrete topology and $A^\mathbb{M}$ the product topology. Cylinders form a countable clopen (open and closed) base of this topology. We consider $A^\mathbb{M}$ to be a metric space with the metric
\[
d(c,e) =
\begin{cases}
 2^{-\min\left(\{\lvert i\rvert\mid c_i\neq e_i\}\right)},&\text{ if }c\neq e \\
0,&\text{ if } c=e
\end{cases},
\]
for all $c,e\in A^\mathbb{M}$. It is well-known that this metric induces the product topology, and that this space is compact.

A \emph{(topological) dynamical system} is a pair $(X,f)$ where $X$ is a compact metric space and $f$ a continuous map $X\rightarrow X$. Let $(X,f)$ and $(Y,g)$ be two dynamical systems. A continuous map $\phi:X\rightarrow Y$ is a \emph{homomorphism} if $\phi f = g\phi$. If $\phi$ is surjective, it is a \emph{factor map}, and $(Y,g)$ is a \emph{factor} of $(X,f)$. If $\phi$ is injective, it is an \emph{embedding}, and $(X,f)$ is a \emph{subsystem} of $(Y,g)$. And lastly, if $\phi$ is a bijection, it is a \emph{conjugacy}, and $(X,f)$ and $(Y,g)$ are \emph{conjugate}, denoted by $(X,f)\cong (Y,g)$. Let $\mathcal{U}$ be a finite open cover of $X$, and denote $h(\mathcal{U})$ the smallest number of elements of $\mathcal{U}$ that cover $X$. Let $\mathcal{V}$ be another finite open cover of $X$ and denote $\mathcal{U}\lor\mathcal{V}=\{U\cap V\mid U\in\mathcal{U},V\in\mathcal{V}\}\setminus\{\varnothing\}$. Then the \emph{entropy of $(X,f)$ with respect to $\mathcal{U}$} is
\[
h(X,f,\mathcal{U})=\lim_{n\rightarrow\infty} \frac{1}{n}\log_2 h\big(\mathcal{U}\lor f^{-1}(\mathcal{U})\lor f^{-2}(\mathcal{U})\lor\cdots \lor f^{-n+1}(\mathcal{U})\big).
\]
The \emph{entropy} of $(X,f)$ is
\[
h(X,f)=\sup\{h(X,f,\mathcal{U})\mid \mathcal{U}\text{ is an open cover of }X\}.
\]
We need the following:

\begin{proposition}{(\cite[Proposition 2.88.]{Kurka03})}
\label{entropy-of-subsystem-and-factor}
If $(Y,g)$ is a subsystem or a factor of $(X,f)$, then $h(Y,g)\leq h(X,f)$. It follows that if $(X,f)$ and $(Y,g)$ are conjugate, then $h(X,f)=h(Y,g)$.
\end{proposition}

The \emph{direct product} of dynamical systems $(X,f)$ and $(Y,g)$ is $(X\times Y, f\times g)$, where $f\times g:X\times Y\rightarrow X\times Y,\;(f\times g)(x,y)=(f(x),g(y))$. It is known that $h(X\times Y,f\times g)=h(X,f)+h(Y,g)$ (\cite[Proposition 2.89]{Kurka03}).

The \emph{shift map} $\sigma:A^\mathbb{M}\rightarrow A^\mathbb{M}$, defined by $\sigma(c)_i=c_{i+1}$ for all $i\in\mathbb{M}$, is easily seen to be continuous. The dynamical system $(A^\mathbb{M},\sigma)$ is the \emph{full ($A$-)shift}. A dynamical system $(X,\sigma)$, where $X\subset A^\mathbb{M}$ is topologically closed and $\sigma^m(X)\subset X$ for all $m\in\mathbb{M}$, is a \emph{subshift}. When it does not cause confusion, we will simply talk about a subshift $X$. A configuration $c\in A^\mathbb{M}$ \emph{avoids} $u\in A^{[0,n)}$ if $\sigma^i(c)_{[0,n)}\neq u$ for all $i\in \mathbb{M}$. Let $S\subseteq A^+$, and let $X_S$ be the set of configurations that avoid $S$, i.e., $X_S=\{c\in A^\mathbb{M}\mid \forall u\in S: c\text{ avoids }u\}$. It is well-known that the given topological definition of subshifts is equivalent to saying that there exists a set of forbidden words $S$ such that $X=X_S$. If there exists a finite set $S$ such that $X=X_S$, then $X$ is a \emph{subshift of finite type (SFT)}. If $Y$ is a factor of an SFT, then it is a \emph{sofic shift}. An equivalent characterization of sofic shifts is that the set of forbidden words is a regular language.

The \emph{subword complexity (of length $n$)} of a subshift $X$ is $p_n(X) =  \lvert \{u\in A^+\mid \exists c\in X: c_{[0,n)}=u\}\rvert$. The entropy of $(X,\sigma)$ can be calculated using the subword complexity
\[
h(X,\sigma)=\lim_{n\rightarrow\infty} \frac{1}{n}\log_2( p_n(X)).
\]

\subsection{Cellular automata}

A \emph{cellular automaton (CA)} is a dynamical system $(X,F)$ where $X\subset A^\mathbb{M}$ is a subshift and $F$ commutes with the shift map, i.e., $F\sigma = \sigma F$. In this paper we will only consider CA's over a full shift, i.e., $X=A^\mathbb{M}$. When $\mathbb{M}=\mathbb{N}$, the CA is called \emph{one-sided} and when $\mathbb{M}=\mathbb{Z}$, the CA is called \emph{two-sided}. We will often refer to a CA by the function alone, i.e., talk about the CA $F$, and in a similar fashion we often omit the phase space from notations, for example write $h(F)=h(A^\mathbb{M},F)$ for the entropy. Let $D=[i,j]\subset\mathbb{M}$ and let $G_l:A^D\rightarrow A$. Define $G:A^\mathbb{M}\rightarrow A^\mathbb{M}$ by $G(c)_i = G_l((\sigma^i(c))_D)$. It is easy to see that $G$ is continuous and commutes with $\sigma$, so it is a cellular automaton. The set $D$ is the \emph{local neighborhood} of $G$ and the function $G_l$ is the \emph{local rule} of $G$. According to the Curtis-Hedlund-Lyndon Theorem every CA is defined by a local rule. We will denote the local and global rules with the same $G$, this will not cause confusion. Let $r\in\mathbb{N}$ be the smallest number such that $D\subseteq[-r,r]$, then $r$ is the \emph{radius} of $G$.

Let $(A^\mathbb{M},F)$ and $(B^\mathbb{M},G)$ be two CA's. If $H:A^\mathbb{M}\rightarrow B^\mathbb{M}$ is a homomorphism from $(A^\mathbb{M},F)$ to $(B^\mathbb{M},G)$, and also a homomorphism from $(A^\mathbb{M},\sigma)$ to $(B^\mathbb{M},\sigma)$ then it is a \emph{strong homomorphism}. Naturally we define \emph{strong factor}, \emph{strong subsystem}, and \emph{strongly conjugate}, when the corresponding homomorphism is a strong homomorphism. If $F$ and $G$ are strongly conjugate, we denote $F\scong G$. Notice that if $\phi$ is a strong conjugacy from $(A^\mathbb{M},F)$ to $(B^\mathbb{M},G)$, then automatically $\phi^{-1}$ is also strong, i.e., commutes with $\sigma$ (see, e.g., \cite{Kari05}).

For every $n\in\mathbb{N}$, CA $(A^\mathbb{M},F)$ defines the \emph{$n$\textsuperscript{th} trace subshift}
\[
\tau_n(F)=\big\{e\in \big(A^{n}\big)^\mathbb{N}\mid \exists c\in A^\mathbb{M}:\forall i\in\mathbb{N}: e_i=\big(F^i(c)\big)_{[0,n)}\big\}.
\]
The entropy of $F$ can be calculated as the limit of the entropies of its trace subshifts
\[
h(F) = \lim_{n\rightarrow\infty} h(\tau_n(F),\sigma).
\]
For a one-sided cellular automaton $F$ with radius $r$ we have that $p_n(\tau_{r+1}(F))=\lvert A\rvert\cdot p_n(\tau_r(F))$, so we get the following:

\begin{proposition}
\label{upper-bound-on-entropy}
Let $F:A^\mathbb{N}\rightarrow A^\mathbb{N}$ be a CA with radius $r$. Then $h(F)=h(\tau_r(F),\sigma)$.
\end{proposition}

Let $F:A^\mathbb{M}\rightarrow A^\mathbb{M}$ and $G:B^\mathbb{M}\rightarrow B^\mathbb{M}$ be two CA's. There are two natural ways to interpret the direct product of $F$ and $G$. First we can consider $F\times G$ to be a CA that has two separate \emph{tracks} $A^\mathbb{M}$ and $B^\mathbb{M}$, and $F\times G$ operates on the $A$-track via $F$ and on the $B$-track via $G$. On the other hand we can also consider $F\times G$ as a CA on $(A\times B)^\mathbb{M}$, where the states have two \emph{layers}. For any $F\times G$ we use which ever interpretation seems more natural, sometimes switching between the two. We can of course define a CA over $(A\times B)^\mathbb{M}$ that is not a direct product of two CA's, for such a CA we will also talk about tracks and layers.

Let $F$ be a CA. If there exist $n,p>0$ such that $F^{n+p}=F^n$, then $F$ is \emph{eventually periodic}, and if there exists $p>0$ such that $F^p=\id$, then $F$ is \emph{periodic}. For a state $a\in A$ we denote ${^\omega a}^\omega\in A^\mathbb{M}$ the configuration such that ${^\omega a}^\omega (i)=a$ for all $i\in \mathbb{M}$. A state $q\in A$ is \emph{quiescent} if $F({^\omega q}^\omega)={^\omega q}^\omega$. A cellular automaton is \emph{nilpotent} if there exists a quiescent state $q$ such that for every $c\in A^\mathbb{M}$ there exists $n\in\mathbb{N}$ such that $F^n(c)={^\omega q}^\omega$. A state $s\in A$ is \emph{spreading} if the local rule maps every neighborhood containing $s$ to $s$. Clearly a spreading state is quiescent. It is known that for cellular automata nilpotency implies uniform nilpotency:

\begin{proposition}{(\cite{CulikPachlYu89})}
\label{uniformlyNilpotent}
Let $F:A^\mathbb{M}\rightarrow A^\mathbb{M}$ be a nilpotent CA. Then there exists $n\in\mathbb{N}$ such that for all $c\in A^\mathbb{M}$ it holds that $F^n(c)={^\omega q}^\omega$.
\end{proposition}

We also need the following, which is a result of a simple compactness argument.

\begin{proposition}
\label{spacetime-diagram-without-spreading-state}
Let $F:A^\mathbb{M}\rightarrow A^\mathbb{M}$ be a CA that is not nilpotent, and let $s\in A$ be a spreading state. Then there exists $c\in A^\mathbb{M}$ such that $F^n(c)_j\neq s$ for all $n\in\mathbb{N}$ and $j\in\mathbb{M}$.
\end{proposition}

Consider a one-sided reversible cellular automaton $F:A^{\mathbb{N}}\rightarrow A^{\mathbb{N}}$ such that both $F$ and its inverse $F^{-1}$ have radius $1$. In many cases this restriction for radius is not a serious one as every reversible CA is conjugate (though maybe not \emph{strongly} conjugate) to such a CA through some grouping of cells. It is easy to see that for every fixed $a\in A$ the map $F(\_a):A\rightarrow A,x\mapsto F(xa)$ has to be a permutation, we will denote this permutation with $\rho_a$. Not every set of permutations $\{\rho_a\}_{a\in A}$ define a reversible CA. We refer the reader to \cite{DartnellMaassSchwartz2003} for a detailed combinatorial considerations of such reversible one-sided CA's. For our purposes the following simple example will be enough.

\begin{example}
\label{021example}
Define a one-sided CA $F:A^\mathbb{N}\rightarrow A^\mathbb{N}$ where $A=\{0,1,2\}$ using the following permuations:
\[
\begin{array}{lr}
&0\mapsto 0\\
\rho_0=\rho_2:&1\mapsto 2\\
&2\mapsto 1\\
\end{array}
\qquad
\begin{array}{lr}
&0\mapsto 1\\
\rho_1:&1\mapsto 2\\
&2\mapsto 0\\
\end{array}.
\]
This is reversible, and its inverse also has radius one, namely the permutations $\pi_0=\pi_1=(0)(12),\pi_2=(021)$ can be verified to define the inverse of $F$. This example was already considered in \cite{DartnellMaassSchwartz2003}. We will compute its entropy.

According to Proposition \ref{upper-bound-on-entropy} the entropy is just the entropy of the subshift $\tau_1(F)$. From the local rule we see that $0$ maps to $0$ or $1$, $1$ always maps to $2$, and $2$ maps to $0$ or $1$. So $\tau_1(F)\subseteq\{0,12\}^\mathbb{Z}$ (which is here considered a subshift of $\{0,1,2\}^\mathbb{Z}$). Suppose $20^n1$ is a factor of some element in $\tau_1(F)$. Notice that the only word of length $n-2$ that can appear next to $20^n1$ in the space-time-diagram of $F$ is $20^{n-2}1$ (consider this with the help of Figure \ref{fig:201}). Inductively this implies that if $20^n1$ is a factor of some element in $\tau_1(F)$ then $n$ is even. So we have that $\tau_1(F)\subseteq\{00,12\}^\mathbb{Z}$. But for any $t\in\{00,12\}^\mathbb{Z}$ we can construct a valid space-time-diagram of $F$ that contains $t$ as follows: Consider $00$ to represent zero and $12$ to represent one, and let $t_1$ be xor of $t$ (turn to Figure \ref{fig:xorlike}). We see that when lined up correctly this gives a configuration that is locally compatible with $t$, i.e., that they could be successive columns of a space-time-diagram of $F$. This process can be repeated to obtain a valid space-time-diagram of $F$.

We have seen that $\tau_1(F)=\{00,12\}^\mathbb{Z}$, and so $h(F)=\frac{1}{2}$. Using the direct product construction we can obtain a one-sided reversible CA that has radius one, and whose inverse also has radius one, and that has arbitrarily high entropy.

\begin{figure}
\centering
\begin{minipage}{.45\textwidth}
  \centering
  \includegraphics{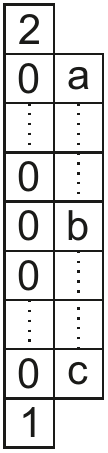}
  \caption{First notice that $c$ has to be $1$, since only $\rho_1$ maps $0$ to $1$. The same way $a$ has to be $2$, since only $\pi_2$ maps $0$ to $2$. Finally $b$ has to be $0$, since $b$ has to satisfy $\rho_b(0)=0$ and $\pi_b(0)=0$.}
  \label{fig:201}
\end{minipage}%
\hfill
\begin{minipage}{.45\textwidth}
  \centering
  \includegraphics{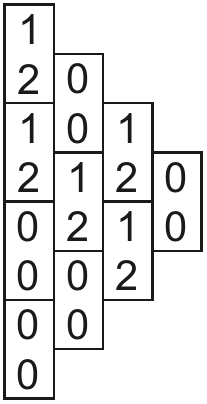}
  \caption{Fill the leftmost column in an arbitrary way using the blocks $00$ and $12$. Fill the next column by taking xor (addition modulo $2$) interpreting $00$ as $0$ and $12$ as $1$. Notice that we get no violations of the local rule of $F$ doing this. Repeat. }
  \label{fig:xorlike}
\end{minipage}
\end{figure}
\end{example}

For an overview of the topics considered here, we refer the reader to \cite{Kari05} (a survey of cellular automata theory), and \cite{Kurka03} (a book on topological and symbolic dynamics).

\section{Main result}

Our proof is based on the undecidability of nilpotency of one-dimensional cellular automata.

\begin{theorem}{(\cite{Kari92},\cite{Lewis1979})}
\label{nilpotency}
Nilpotency of one-dimensional one-sided cellular automata with a spreading state and radius $1$ is undecidable.
\end{theorem}

Now we can prove the main result of this paper.

\begin{theorem}
\label{inseparablePairs}
The following two sets of pairs of one-dimensional one-sided cellular automata are recursively inseparable:
\begin{enumerate}[(i)]
\item pairs where the first cellular automaton has strictly higher entropy than the second one, and
\item pairs that are strongly conjugate and both have zero topological entropy.
\end{enumerate}
\end{theorem}
\begin{proof}
We will reduce the decision problem of Theorem \ref{nilpotency} to this problem, which will prove our claim.

Let $H:B^\mathbb{N}\rightarrow B^\mathbb{N}$ be an arbitrary given one-sided CA with neighborhood radius $1$ and a spreading quiescent state $q\in B$. Let $k\in\mathbb{N}$ be such that $k>\log_2(\lvert B\rvert)$, $F_{2k}$ be the $2k$-fold cartesian product of the cellular automaton $F$ of Example \ref{021example}, and $A=\{0,1,2\}^{2k}$ (we are aiming for high enough entropy). Now we are ready to define CA's $\mathcal{F}$ and $\mathcal{G}$ such that
\begin{align*}
H\text{ is not nilpotent}&\implies h(\mathcal{F})>h(\mathcal{G})\\
H\text{ is nilpotent}&\implies \mathcal{F}\scong \mathcal{G}\text{ and }h(\mathcal{F})=h(\mathcal{G})=0.
\end{align*}
Both of these new CA's work on two tracks $\mathcal{F},\mathcal{G}:\left(A\times B\right)^\mathbb{N}\rightarrow \left(A\times B\right)^\mathbb{N}$. The CA $\mathcal{G}$ is simply $\id_A\times H$, i.e.,
\[
\mathcal{G}((a_0,b_0)(a_1,b_1))=(a_0,H(b_0 b_1)),
\]
for all $a_0,a_1\in A,\;b_0,b_1\in B$. The CA $\mathcal{F}$ acts on the $A$-track as $F_{2n}$ when the $B$-track is not going to become $q$, and as $\id_A$ when the $B$-track is going to become $q$, i.e., 
\[
\mathcal{F}((a_0,b_0)(a_1,b_1))=
\begin{cases}
(F_{2n}(a_0a_1),H(b_0b_1)),&\text{if } H(b_0b_1)\neq q\\
(a_0,H(b_0b_1)),&\text{if } H(b_0b_1)=q,
\end{cases}
\]
for all $a_0,a_1\in A,\;b_0,b_1\in B$.

\textbf{(i)} Suppose that $H$ is not nilpotent. The entropy of $\mathcal{G}$ is
\[
h\big(\left(A\times B\right)^\mathbb{N},\mathcal{G}\big)=h\left(A^\mathbb{N},\id_A\right)+h\left(B^\mathbb{N},H\right)=h\left(B^\mathbb{N},H\right),
\]
since $\mathcal{G}=\id_A\times H$. On the other hand, by Proposition \ref{spacetime-diagram-without-spreading-state}, there exists a configuration $e\in B^\mathbb{Z}$ such that for all $i,j\in\mathbb{N}$ we have that $H^i(c)_j\neq q$. But then we have that
\[
h\big(\left(A\times B\right)^\mathbb{N},\mathcal{F}\big)\geq h\left(A^\mathbb{N},F_{2k}\right)>\log_2( \lvert B\rvert)\geq h\big(B^\mathbb{N},H\big),
\]
according to Example \ref{021example} and how we chose $k$. Overall we have that
\[
h\big(\left(A\times B\right)^\mathbb{N},\mathcal{F}\big)>h\big(\left(A\times B\right)^\mathbb{N},\mathcal{G}\big),
\]
as was claimed.

\textbf{(ii)} Suppose that $H$ is nilpotent. Let us first explain informally why we now have that $\mathcal{F}\scong \mathcal{G}$. Both $\mathcal{F}$ and $\mathcal{G}$ behave identically on the $B$-track, so the conjugacy will map this layer simply by identity. Nilpotency of $H$ guarantees that for all configurations the $B$-track will be ${^\omega q}^\omega$ after some constant time $n$. By the definition of $\mathcal{F}$ this means that after $n$ steps $\mathcal{F}$ does nothing on the $A$-track. Since $\mathcal{G}$ never does anything on the $A$-track, we can use this fact to define the conjugacy on the $A$-track simply with $\mathcal{F}^n$. That this is in fact a conjugacy follows since $\mathcal{F}$ is, informally, reversible on the $A$-layer for a fixed $B$-layer.

Let us be exact. First we will define a continuous map $\phi:(A\times B)^\mathbb{N}\rightarrow (A\times B)^\mathbb{N}$ such that $\phi \mathcal{F} =\mathcal{G} \phi$. This $\phi$ will be a CA. Then we show that $\phi$ is injective, which implies reversibility (see, e.g., \cite{Kari05}), and so $\mathcal{F}\scong \mathcal{G}$.

Let $\pi_A:A^\mathbb{N}\times B^\mathbb{N}\rightarrow A^\mathbb{N}$ be the projection $\pi_A(c,e)=c$ for all $c\in A^\mathbb{N}$ and $e\in B^\mathbb{N}$. Define $\pi_B:A^\mathbb{N}\times B^\mathbb{N}\rightarrow B^\mathbb{N}$ similarly.

Let $n\in\mathbb{N}$ be a number such that for all $c\in B^\mathbb{N}$ we have $H^n(c)={^\omega q}^\omega$. Such $n$ exists according to Proposition \ref{uniformlyNilpotent}, since $H$ is nilpotent. Because $\mathcal{F}$ and $\mathcal{G}$ act identically on the $B$-track, $\phi$ will map this layer simply by identity, i.e.,
\[
\pi_B \phi(c,e)=e,
\]
for all $c\in A^\mathbb{N},e\in B^\mathbb{N}$. On the $A$-layer $\phi$ is defined using the fact that after $n$ steps $\mathcal{F}$ does nothing on the $A$-track, i.e., acts the same way $\mathcal{G}$ does. Due to this we define
\[
\pi_A\phi = \pi_A {\mathcal{F}}^n.
\]
Now $\phi$ is a CA, since it is continuous and shift-commuting. Let us show that $\phi$ is a homomorphism. Of course we have that
\[
\phi \mathcal{F} = \mathcal{G} \phi \iff \big(\pi_A\phi \mathcal{F} = \pi_A \mathcal{G}\phi\text{ and } \pi_B\phi \mathcal{F} = \pi_B \mathcal{G}\phi\big).
\]
It is immediate from the definitions that $\pi_B\phi \mathcal{F} = \pi_B \mathcal{G}\phi$. For the equality on the $A$-layer notice first that $\pi_A\mathcal{G}=\pi_A$, and then compute:
\begin{align*}
\pi_A\phi \mathcal{F} \;&\stackrel{\mathmakebox[\widthof{=}]{\text{def.}}}{=}\; (\pi_A {\mathcal{F}}^n)\mathcal{F}\\
&=\;\pi_A \mathcal{F}{\mathcal{F}}^n & &\mid\mid\text{ after }n\text{ steps } \mathcal{F}\\
&=\;\pi_A \mathcal{G} {\mathcal{F}}^n & &\quad\;\;\text{behaves as }\mathcal{G}\\
&=\;\pi_A {\mathcal{F}}^n\\
&\stackrel{\mathmakebox[\widthof{=}]{\text{def.}}}{=}\;\pi_A\phi\\
&=\;\pi_A \mathcal{G}\phi.
\end{align*}
So we have that $\phi \mathcal{F} = \mathcal{G} \phi$.

To prove that $\phi$ is a strong \emph{conjugacy} it is enough to show that $\phi$ is an injection. As the $B$-layer is mapped by identity, we only need to show that for a fixed $e\in B^\mathbb{N}$ we have that for all $c\in A^\mathbb{N}$ there exists a unique $c'\in A^\mathbb{N}$ such that $\phi(c',e)=(c,e)$. By the definition of $\phi$ it is clear that this will hold if
\begin{align*}
\pi_A {\mathcal{F}}^n(\_,e):A^\mathbb{N}&\longrightarrow A^\mathbb{N}\\
c\;\;&\longmapsto \pi_A {\mathcal{F}}^n(c,e)
\end{align*}
is a bijection for every $e\in B^\mathbb{N}$. We can consider this step by step. We claim that $(c,e)=(c_0c_1c_2\dots,e_0e_1e_2\dots)\in(A\times B)^\mathbb{N}$ uniquely defines the $A$-track of the elements in the set $\mathcal{F}^{-1}(c,e)$. Let $(c',e')=(c_0'c_1'c_2'\dots,e_0'e_1'e_2'\dots)\in\mathcal{F}^{-1}(c,e)$. It is enough to show that $c_0'$ is defined uniquely by $(c,e)$. Suppose first that $e_0=q$. Then according to the definition $\mathcal{F}$ acted as identity, so we have that $c_0'=c_0$. Suppose next that $e_0\neq q$. We have two cases, either $e_1=q$ or not. Suppose first that $e_1=q$. Then as before we have that $c_1'=c_1$. And so $c_0'=\rho_{c_1'}^{-1}(c_0)=\rho_{c_1}^{-1}(c_0)$. And lastly suppose that $e_1\neq q$. Then we have that $F_{2n}(c_0'c_1'c_2'\dots)=(c_0c_1\dots)$ according to the definition of $\mathcal{F}$. But now $c_0'$ is uniquely determined since $F_{2n}$ is reversible and the inverse also has radius $1$, namely we have that $c_0'=F_{2n}^{-1}(c_0c_1)$.

To complete the proof we observe that
\[
h(\mathcal{F})=h(\mathcal{G})=h(\id_A)+h(H)=0,
\]
since $\mathcal{F}\scong \mathcal{G}=\id_A\times H$, and $H$ is nilpotent.
\end{proof}

Since the two-sided variant of Theorem \ref{inseparablePairs} can be reduced to the one-sided case, also the two-sided variant is undecidable. We also get the following corollary.

\begin{corollary}
\label{undecidableProperties}
Let $\mathbb{M}=\mathbb{N}$ or $\mathbb{M}=\mathbb{Z}$. Let $F,G:A^\mathbb{M}\rightarrow A^\mathbb{M}$ be two cellular automata. Then the following hold:
\begin{enumerate}
\item It is undecidable whether $F$ and $G$ are (strongly) conjugate.
\item It is undecidable whether $F$ is a (strong) factor of $G$.
\item It is undecidable whether $F$ is a (strong) subsystem of $G$.
\end{enumerate}
\end{corollary}

\begin{proof}
{\it 1.} The pairs in the set {\it (i)} of Theorem \ref{inseparablePairs} can not be (strongly) conjugate, and the pairs in {\it(ii)} have to be. Thus deciding (strong) conjugacy would separate these sets.

{\it 2.} One of the CA's in the pair from the set {\it (i)} has strictly higher entropy than the other, so it can not be a (strong) factor of the other. On the other hand CA's of pairs from the set {\it (ii)} are (strong) factors of each other. So checking whether both CA's of a pair is a (strong) factor of the other would separate the sets of Theorem \ref{inseparablePairs}.

{\it 3.} In a similar way, since a subsystem can not have higher entropy.
\end{proof}

\section{Other results}

\subsection{Decidable cases}

Now that we know conjugacy to be undecidable for one-dimensional cellular automata, we can consider what happens if we restrict to some natural subclass. Recently it was proved that

\begin{theorem}{(\cite[Corollary 5.17.]{Epperlein2017})}
Conjugacy of periodic cellular automata on one- or two-sided subshifts of finite type is decidable.
\end{theorem}

Periodic cellular automata are the least sensitive to changes in the initial configuration. Next we consider the most sensitive cellular automata, i.e., positively expansive ones. A dynamical system $(X,f)$ is called \emph{positively expansive} if
\[
\exists\varepsilon>0:\forall x,y\in X:\exists n\in\mathbb{N}: x\neq y\implies d(f^n(x),f^n(y))>\varepsilon.
\]
Positively expansive CA's are quite extensively studied which allows us to deduce the following result.

\begin{proposition}
Conjugacy of positively expansive cellular automata on one- or two-sided full shifts is decidable.
\end{proposition}
\begin{proof}
Let $F:A^\mathbb{M}\rightarrow A^\mathbb{M}$ and $G:B^\mathbb{M}\rightarrow B^\mathbb{M}$ be two positively expansive cellular automata. Due to the positive expansivity, $F$ and $G$ are conjugate to $\tau_k(F)$ and $\tau_k(G)$ (resp.) for large enough $k$. These subshifts are conjugate to subshifts of finite type (\cite{BoyleKitchens99} for one-sided case, \cite{Nasu95} for two-sided case). According to \cite[Theorem 36]{diLena2007} we can effectively compute these subshifts. The claim follows, as the conjugacy of one-sided subshifts of finite type is decidable (\cite{Williams73}).
\end{proof}

Naturally these results raise the question whether strong conjugacy is decidable when restricted to periodic or positively expansive cellular automata. Also the questions whether (strong) conjugacy is decidable for eventually periodic, i.e., equicontinuous, cellular automata (\cite[Question 8.1.]{Epperlein2017}), or for expansive cellular automata remain unanswered. It is conjectured that expansive cellular automata are conjugate to two-sided subshifts of finite type (this is known for expansive two-sided cellular automata with one-sided neighborhoods), however the previous proof still wouldn't work, as it is not known whether conjugacy of two-sided subshifts of finite type is decidable.

\subsection{Conjugacy of subshifts}

Questions about conjugacy provide perhaps the most well-known open problems in symbolic dynamics. For example it is unknown whether conjugacy of two-sided subshifts of finite type is decidable. It is also unknown whether conjugacy of one- or two-sided sofic shifts is decidable. On the other hand conjugacy of one-sided subshifts of finite type is known to be decidable; we used this fact to show that conjugacy of positively expansive cellular automata is decidable. We can ask if we could work to the opposite direction, i.e., if the classical problems for subshifts could be answered using cellular automata. For example, undecidability of conjugacy for one-sided expansive cellular automata would imply undecidability of conjugacy of two-sided subshifts of finite type, although it seems more likely that conjugacy for one-sided expansive cellular automata is decidable. A more plausible result would be that conjugacy is undecidable for expansive two-sided cellular automata, which together with the conjecture that every expansive cellular automaton is conjugated to a two-sided SFT (\cite[Conjecture 30.]{Kurka09}), would imply undecidability of conjugacy of two-sided SFT's.

All of the above relied on the connection between cellular automaton and its trace subshift. The problem with this approach is that only expansive cellular automata are conjugate to subshifts. However there could be some more inventive ways to link subshifts and cellular automata to obtain decidability results. We provide the following, somewhat artificial, result.

\begin{proposition}
Let $X,Y\subseteq(A\times A)^\mathbb{M}$ be two subshifts of finite type. It is undecidable whether $X$ and $Y$ are conjugate via a conjugacy of the form $\phi\times\phi$.
\end{proposition}
\begin{proof}
The proof is a direct reduction from strong conjugacy of cellular automata. Let $F,G:A^\mathbb{M}\rightarrow A^\mathbb{M}$ be two CA's. Let $X=\{(c,F(c))\mid c\in A^\mathbb{M}\}$ and $Y=\{(c,G(c))\mid c\in A^\mathbb{M}\}$. These subshifts are naturally conjugate to $A^\mathbb{M}$. Suppose there exists a conjugacy $\phi\times\phi:X\rightarrow Y$. Then $\phi$ commutes with the shift and for every $c\in A^\mathbb{M}$ we have that $(\phi(c),\phi F(c))=(e,G(e))$, where  $e$ has to be $\phi(c)$, and so $\phi F(c)=G\phi(c)$ for all $c\in A^\mathbb{M}$. In other words $\phi$ is a strong conjugacy of $(A^\mathbb{M},F)$ and $(A^\mathbb{M},G)$.

On the other hand, any strong conjugacy $\phi$ from $(A^\mathbb{M},F)$ to $(A^\mathbb{M},G)$ immediately gives a conjugacy $\phi\times\phi$ between $X$ and $Y$.
\end{proof}

\section{Conclusion}

We have proved that the decision problems "are (strongly) conjugate", "is a (strong) subsystem of" and "is a (strong) factor of" are undecidable for one-dimensional one- and two-sided cellular automata. We note that these results provide an example that contradicts the rule of thumb that one time step properties of one-dimensional cellular automata are decidable.

A natural question to ask is whether conjugacy remains undecidable even for reversible cellular automata. Since our proof is based on the undecidability of nilpotency, it is clear that a different approach is needed. We note that though for non-reversible cellular automata  one- and two-sided cases differ only little, for reversible cellular automata the one-sided case seems far more distant as there are no known undecidability results for one-sided cellular automata that could be used for the reduction. For two-sided cellular automata periodicity and mortality problems (\cite{KariOllinger2008},\cite{KariLukkarila2009}) are known to be undecidable, and provide a possible replacement for the nilpotency problem in the reversible case. This is of course implicitly assuming that one is expecting the problem to remain undecidable.

Lastly it is interesting to consider whether there is way to solve or at least shed new light on the long-standing open problems of symbolic dynamics, namely conjugacy problems of subshifts.

\end{document}